\documentclass[]{interact}

\usepackage{epstopdf}
\usepackage[caption=false]{subfig}

\usepackage[natbibapa,nodoi]{apacite}
\usepackage{amsmath}
\usepackage{amsthm}
\usepackage{amsfonts}
\usepackage{amssymb}
\usepackage{bm}
\usepackage{xcolor}
\usepackage{algorithm}
\usepackage{algorithmic}
\usepackage{natbib}
\usepackage{graphicx}
\usepackage{comment}
\usepackage{multirow}
\usepackage{booktabs}
\usepackage{siunitx}
\usepackage{accents}

\usepackage[capitalize]{cleveref}
\crefname{section}{Sec.}{Secs.}
\Crefname{section}{Section}{Sections}
\Crefname{table}{Table}{Tables}
\crefname{table}{Tab.}{Tabs.}

\DeclareMathOperator{\sign}{sign}

\newcommand{\indep}{\perp \!\!\! \perp}

\def\Ncal{{\mathcal{N}}}

\def\Rbb{{\mathbb{R}}}

\def\xbm{{\bm{x}}}

\def\Ibf{{\mathbf{I}}}

\def\Expect{\mathrm{E}}
\def\sign{\mathrm{sign}}
\def\iter{\mathrm{iter}}

\def\dCov{\mathrm{dCov}}

\def\dCor{\mathrm{dCor}}

\def\Xb{{\bf X}}

\def\bC{{\bf C}}

\def\bU{{\bf U}}
\def\bV{{\bf V}}
\def\bQ{{\bf Q}}

\def\di{{\mathrm d}}

\def\bz{{\bf z}}

\def\Xb{{\bf X}}
\def\bX{{\bf X}}

\def\bY{{\bf Y}}
\def\bI{{\bf I}}

\def\bZ{{\bf Z}}

\def\eps{{\epsilon}}

\def\reals{{\mathbb R}}

\def\Dcal{\mathcal{D}}

\theoremstyle{plain}
\newtheorem{theorem}{Theorem}[section]

\newtheorem{proposition}[theorem]{Proposition}

\theoremstyle{definition}

\theoremstyle{remark}

\begin{document}

\articletype{Journal of Nonparametric Statistics}

\title{Robust Sufficient Dimension Reduction via $\alpha$-Distance Covariance
}

\author{
\name{Hsin-Hsiung Huang\textsuperscript{a} 
\thanks{CONTACT Hsin-Hsiung Huang. Email: hsin.huang@ucf.edu}
and Feng Yu\textsuperscript{b} and Teng Zhang\textsuperscript{c}\thanks{CONTACT Teng Zhang. Email: teng.zhang@ucf.edu}}
\affil{
\textsuperscript{a} Department of Statistics and Data Science, University of Central Florida, USA; \textsuperscript{b} Department of Mathematics, University of Minnesota Twin Cities, Minnesota, USA; 
\textsuperscript{c} Department of Mathematics, University of Central Florida, USA}
}

\maketitle

\begin{abstract}
We introduce a novel sufficient dimension-reduction (SDR) method which is robust against outliers using $\alpha$-distance covariance (dCov) in dimension-reduction problems. Under very mild conditions on the predictors, the central subspace is effectively estimated and model-free without estimating link function based on the projection on the Stiefel manifold. We establish the convergence property of the proposed estimation under some regularity conditions. We compare the performance of our method with existing SDR methods by simulation and real data analysis and show that our algorithm improves the computational efficiency and effectiveness. 
\end{abstract}


\section{Introduction}

In regression analysis, sufficient dimension reduction (SDR) provides a useful statistical framework to analyze a
high-dimensional dataset without losing any information. It finds the fewest linear combinations of predictors that
capture a full regression relationship. Let $Y$ be an univariate response and $X = (x_1,\ldots,x_p)^T$ be a $p\times 1$ predictor
vector, SDR aims to find a $p\times d$ matrix $\bm{\beta}$ such that
$$
Y \indep X\mid \bm{\beta}^TX
$$
which denotes the statistical independence. 

Sufficient dimension reduction (SDR) based on the conditional distribution of the response \citep{Li1991,cook1991discussion,xia2002adaptive,Yin11-AOS950}
provides the reduced predictors without loss of regression information. Recently, SDR methods using distance covariance (dCov) have been developed \citep{sheng2013direction,sheng2016sufficient}, and such methods do not need a constant covariance condition, or distribution assumptions on $X$,
$X\mid Y$ or $Y \mid X$. Therefore, it has broad applications for continuous and discrete variables from various distributions. Several robust sufficient dimension reduction methods have proposed for coefficient estimation such as the robust sufficient dimension reduction using the ball covariance \citep{zhang2019robust} and the
expected likelihood based method that minimizes the Kullback-Leiblier distance
\citep{yin2005direction,zhang2015direction}. 
In this article, we propose a robust estimation of sufficient dimension reduction the independence via the $\alpha$-distance covariance ($\alpha$-dCov) between the response and the predictors and develop a new
algorithm for estimating directions in general multiple-index models with a form 
$$
Y = g(\bm{\beta}^T X, \epsilon), 
$$ where $g$ is an unknown link function
\citep{yin2008successive,xia2008multiple,sheng2013direction}. The rest of this article is organized
as follows: Section 2 describes our robust $\alpha$-dCov method and a corresponding outlier detection method, including motivation, theoretical
results, estimation algorithm, and testing procedure. We introduce the consistency theorem in Section 3. Section 4 contains simulation and real data studies. We summarize our work in Section 5.

\subsection{Generalized distance covariance}\label{GDCOV}

Distance covariance \citep{10.1214/009053607000000505} is a popular dependence measure for two random vectors of possibly different dimensions and types. In recent years, there have been concentrated efforts in the literature to understand the distributional properties of the sample distance covariance in a high-dimensional setting, with an exclusive emphasis on the null case that $X$ and $Y$ are independent.
Distance covariance can be generalized to include powers of Euclidean distance. Define
\begin{align}\label{eq:nu}
\nu^{2}(X,Y;\alpha )&:=\mbox{E} [\|X-X'\|^{\alpha }\,\|Y-Y'\|^{\alpha }]+\mbox{E} [\|X-X'\|^{\alpha }]\,\mbox{E} [\|Y-Y'\|^{\alpha }]\nonumber\\&-2\mbox{E} [\|X-X'\|^{\alpha }\,\|Y-Y''\|^{\alpha }],
\end{align}
where $(X,Y),\; (X', Y'),\; (X'',Y'')$ are independent and identically distributed (i.i.d.) with respect to the joint distribution of $(X,Y)$ \citep{10.1214/14-AOS1255}. As discussed in \citep{10.1214/009053607000000505}, for every $0<\alpha <2$, $X$ and $Y$ are independent if and only if $\nu^{2}(X,Y;\alpha )=0$. When $\alpha=1$, it reduces to the classical distance covariance. When $0<\alpha<1$, it can be considered as a more robust version of distance covariance as it reduces the influence of large values of $\|X-X'\|$, $\|Y-Y'\|$, and $\|Y-Y''\|$ that might be contributed to outliers.

\subsection{Central Space Estimation via $\alpha$-dCov}

Let $(\bX, \bY) = \{
(X_i,Y_i):\; i = 1,\ldots,n\}$
be $n$ random samples from random variables $(X, Y)$. In addition, $\bX$ denotes a $p\times n$ data matrix whose columns are $X_1,\cdots,X_n$ and $\bY=[Y_1,\cdots,Y_n]$ denotes
a $1 \times n$ response data matrix. In this article, we consider univariate responses. However, the method can naturally be extended
to multivariate responses without any issue due to the nature of $\alpha$-dCov. The empirical solution of the SDR method based on $\alpha$-dCov for these $n$ observations relies on solving the following objective function \citep{10.1214/009053607000000505,sheng2016sufficient}:
\begin{align}
\max_{\bm{\beta}\in\mathbb{R}^{p\times d}}\nu^2_n(\bm{\beta}^T \bX,\bY,\alpha).
\label{eq1}
\end{align}
with constraint $\bm{\beta}^T\Sigma_{X}\bm{\beta}=\bI_{d}$ and $1\leq d\leq p$, where $\nu_n$ is the empirical version of $\nu$ defined in Equation~\eqref{eq:nu}.

The empirical distance dependence statistics $\nu_n$ is defined as follows. For $k,l = 1, \ldots , n$, we compute the Euclidean
distance matrices $(a_{kl}) = (|X_k - X_l |^{\alpha}_p)$ and $(b_{kl}) = (|Y_k -Y_l |^{\alpha})$ for $0<\alpha <2$ \citep{szekely2009brownian}. Define
$$
A_{kl} = a_{kl} - \bar{a}_{k\cdot} - \bar{a}_{\cdot l} + \bar{a}_{\cdot\cdot},\quad k,l = 1, \ldots , n,
$$
where
$$
\bar{a}_{k\cdot} = \frac{1}{n}\sum^n_{l=1}a_{kl},\quad
\bar{a}_{\cdot l},= \frac{1}{n}\sum^n_{k=1}a_{kl},\quad
\bar{a}_{\cdot\cdot} = \frac{1}{n^2}\sum^n_{k,l=1}a_{kl}.
$$
Similarly, define $B_{kl} = b_{kl} - \bar{b}_{k\cdot} - \bar{b}_{\cdot l} + \bar{b}_{\cdot\cdot}$, for $k, l = 1, \ldots , n$.
The nonnegative sample distance covariance $\nu_n(\bX,\bY)$ and
sample distance correlation $R_n(\bX,\bY)$ are defined by
\begin{equation}
\nu^2_n(\bm{\beta}^T\bX,\bY,\alpha) = \frac{1}{n^2}\sum^n_{k,l=1} A_{kl}B_{kl}
\label{def:alpha_dCOV}
\end{equation}
and
\begin{equation*}
  R^2_n(\bX,\bY,\alpha) =
    \begin{cases}
\frac{\nu^2_n(\bX,\bY,\alpha)}{\nu^2_n(\bX,\alpha)\nu^2_n(\bY,\alpha)},\;&
\mbox{ if }\nu^2_n(\bX,\alpha)\nu^2_n(\bY,\alpha) > 0;\\
      0 ,\;&
\mbox{ if }\nu^2_n(\bX,\alpha)\nu^2_n(\bY,\alpha) = 0,
    \end{cases}       
\end{equation*}
respectively, where the sample distance variance is defined by
$$
\nu^2_n(\bX,\alpha) := \nu^2_n(\bX,\bX,\alpha) = \frac{1}{n^2}
\sum^n_{k,l=1}
A^2_{kl} .
$$

Following \cite{wu2021mm}, we have the following equivalence.
Let $\bC = \hat{\Sigma}^{\frac12}_X \bm{\beta}$ and $\bZ = \hat{\Sigma}^{-\frac12}_X \bX,$ the target function~\eqref{eq1} can be rewritten as
\begin{align}\label{eq:dcov}
    \max_{\bC} \nu^2_n (\bC^T \bZ, \bY,\alpha) := \frac{1}{n^2}\sum^n_{k,l=1}a_{kl}(\bC)B_{kl},\mbox{ s.t. }\bC \in \mbox{St}(d, p),
\end{align}
where $a_{kl}(\bC) = \|\bC^T Z_k - \bC^T Z_l \|^{\alpha}$. We use the same notation 
$\mbox{St}(d, p) = \{\bC \in \mathbb{R}^{p\times d}\mid \bC^T\bC = I_d\}$ with $d \leq p$ is referred to the Stiefel
manifold and $T_{\bC} \mbox{St}(d, p)$ is the tangent space to $\mbox{St}(d, p)$ at a point $\bC\in \mbox{St}(d, p)$. 
We assume that $Y = g(\bC^T \bZ , \epsilon)$, where $\bC$ is a $p\times d$
matrix, $\epsilon$ is an unknown random error independent of $\bZ$, and $g$ is an unknown link function. We propose a new method to
estimate a basis of the central subspace $S_{Y \mid \bZ }= \mbox{Span}(\bC)$ and denote $\nu^2_n (\bC^T \bZ, \bY,\alpha)$ as $ F(\bC)$.

\section{Algorithm}
We develop an iterative algorithm based on the gradient descent algorithm on the Stiefel manifold. Here $P_S$ is a projection on the Stiefel manifold \citep{dalmau2017projection}.
By Proposition 3.4 (the projection onto Stiefel manifolds) of \cite{absil2012projection}, we
let $\bar{\bC} \in \mbox{St}(d,p)$ for any $\bC$
such that $\|X - \bar{X}\| < \sigma_d(\bar{\bC})$, where $\sigma_d(\bar{\bC})$ is the largest singular value, then the projection of $\bC$ onto $\mbox{St}(d,p)$ exists uniquely, and can be expressed as
$P_S(\bC) = \sum^d_{i=1}
u_iv_i^T$,
given by a singular value decomposition of $\bC$. Alternatively, let the SVD of $\bC\in\mathbb{R}^{p\times d}$ be $\bC=\bU\Sigma\bV$, then $P_S(\bC)=\bU\bV^T$.

\begin{algorithm}
  \caption{rSDR: robust SDR}
  \label{alg:matrix-reg-log}
  \begin{algorithmic}[1]
    \STATE {\bf Input:}  The samples $\{(y_i,\bZ_i), i=1,\cdots,n\}$, initial $\bC^{(0)}$.
    \STATE {\bf Initialization:} $\bC^{(0)}$.
    \FOR{$\iter=0,1,\cdots$}
        \STATE Let $\bC^{(\iter+1)}=P_S(\bC^{(\iter)}+\alpha^{(\iter)}_1 \partial_\bC F(\bC^{(\iter)}))$ or $\bC^{(\iter+1)}=P_S\Big(\bC^{(\iter)}+\alpha^{(\iter)}_1 \partial_\bC F(\bC^{(\iter)})(\bI-\bC^{(\iter)\,T}\bC^{(\iter)})\Big)$, where  $P_S(\cdot)$ is the projection on the Stiefel manifold and $\frac{\partial}{\partial \bC} F(\bC)$, and $\alpha^{(\iter)}_1$ is chosen by a line search. 
        \STATE Repeat steps 4 until $\|F(\bC^{(\iter)})-F(\bC^{(\iter-1)})\|_F \leq \eps_n$ where $\eps_n$ is a pre-specified threshold, or the number of iterations exceeds the upper limit: $\iter > N^{(\max)}$.
    \ENDFOR
    \STATE {\bf Output:} Estimated coefficients $\hat{\bC}$.
  \end{algorithmic}
\end{algorithm}

Now we derive the explicit formula for $\partial_\bC F(\bC)$, where $F(\bC)=\nu^2_n(\bC^T\bZ,Y)=\frac{1}{n^2}
\sum^n_{k,l=1}
a_{kl}(\bC)B_{kl}$. Recall that $a_{kl}(\bC)=\|\bC^T \bZ_k-\bC^T\bZ_l\|^{\alpha}$, the gradient is
\begin{equation}\label{eq:gradient}
\partial_\bC F(\bC)=\frac{\alpha}{n^2}
\sum^n_{k,l=1} \frac{\bC^T(\bZ_k-\bZ_l)(\bZ_k-\bZ_l)^T}{\|\bC^T \bZ_k-\bC^T \bZ_l\|^{2-\alpha}}B_{kl},
\end{equation}
and one may perform the manifold gradient descent algorithm as follows: 
\[
\bC^{(\iter+1)}=P_S\Big(\bC^{(\iter)}+\alpha^{(\iter)}_1 \partial_\bC F(\bC^{(\iter)})(\bI-\bC^{(\iter)}\bC^{(\iter)\,T})\Big).
\]
We remark that while there are various advanced Stiefel manifold optimization algorithms such as the ones based on the Cayley transform \citep{wen2012,zhu2019orthodr} or geodesics  \citep{absil2009optimization}, we applied the standard projected gradient descent algorithm as it is simpler to implementation and has the same order of computational cost per iteration of $O(p^2d)$. 

\textbf{Implementation issues} When implementing our approach, practical challenges may arise due to the potential for an extremely small denominator in Equation~\eqref{eq:gradient}, disproportionately amplifying the influence of the $(k,l)$-th term. To preemptively address this concern, we introduce a small positive regularization parameter, denoted as $\eta$. Subsequently, we employ a regularization technique on the objective function $F(\bC)$ such that the $(k,l)$-th term of the gradient in Equation~\eqref{eq:gradient} remains bounded. In particular, we apply it to the regularized objective function denoted as $F_{\eta}(\bC)$:
\[F_{\eta}(\bC)=\frac{1}{n^2}(\|\bC^T Z_k - \bC^T Z_l \|^2+\eta)^{\alpha/2}B_{kl},
\]
which leads to the regularized gradient formulation expressed as follows:
\begin{equation}\label{eq:gradient1}
\frac{\alpha}{n^2}
\partial_\bC F_{\eta}(\bC)=\sum^n_{k,l=1} \frac{\bC^T(\bZ_k-\bZ_l)(\bZ_k-\bZ_l)^T}{(\|\bC^T \bZ_k-\bC^T \bZ_l\|^2+\eta)^{(2-\alpha)/2}}B_{kl}.
\end{equation}

\section{Consistency Theory}
We consider a  model with a general noise term
\[\bY=g(\bm{\beta}_0^T\bX,\epsilon)=g(\bC_0^T\bZ,\epsilon),\] 
where $\bm{\beta}_0$ is a $p\times d$ orthogonal matrix, $g(\cdot)$ is an unknown link function, $\bC_0 = \hat{\Sigma}^{\frac12}_X \bm{\beta}_0$, and $\bZ = \hat{\Sigma}^{-\frac12}_X \bX$, and $\epsilon$ is independent of $\bZ$. This model includes the model from \cite{xia2002adaptive} that
\[\bY=g(\bm{\beta}_0^T \bX)+\epsilon\]
as a special example.

Following \cite{sheng2016sufficient},
we have the asymptotic properties of the estimator $\hat{\bC}$ that is consistent. The statement and the proof is similar to that of \cite{sheng2013direction}. It requires an additional assumption that depends on the decomposition of $X$ into two independent components, and some discussions on this condition are available in \cite[Section 3.2]{sheng2013direction}. For example, it is satisfied when $X$ is normal \citep{zhang2015direction}. In addition, this assumption also holds asymptotically 
 when $p$ is large \citep{hall1993almost}.

The following proposition establishes the asymptotic properties of our estimator $\bC$ up to some rotation matrix $\bQ$. This implies the asymptotic property of the estimated central subspace as it is invariant to the rotation matrix.
 
\begin{proposition}
Let $\bC \in \mathbb{R}^{d\times p}$ be a basis of the central subspace $S_{Y\mid X}$ with $\bC^T\Sigma_X\bC=I_d$. 
Suppose $P^T_{\bC (\Sigma_X)} X \indep Q^T_{\bC(\Sigma_X )}X$ 
and the support of $X \in \mathbb{R}^{d\times p}$, say $S$, is a compact set. In addition, assume that there exists $\bC'\in\mathbb{R}^{(p-d)\times p}$ such that $[\bC,\bC']^T\Sigma_X[\bC,\bC']=I_p$ and $\bC^TX$ is independent of $\bC'^TX$.
Let $\hat{\bC}=\arg\min_{\bC^T\Sigma_X\bC=I_d}\nu^2_n(\bC^T\bX,\bY)$, then there exists a rotation matrix $\bQ$: $\bQ^T\bQ=I_d$ such that
$\hat{\bC}\stackrel{P}{\to}\bC\bQ$ (convergence in probability) as $n \to\infty$.
\label{prop1}
\end{proposition}

\begin{proof}
Following \cite[(4.1)]{szekely2009brownian}, we have that for random variables $X$ and $Y$ from $\reals^{p_1}$ and $\reals^{p_2}$,
\[
\nu^2(X,Y,\alpha)=C\int_{t,s}\frac{|f_{X,Y}(t,s)-f_X(f)f_Y(t)|^2}{\|t\|^{p_1+\alpha}\|s\|^{p_2+\alpha}}\di t\di s,
\] 
where $f_{X}$, $f_{Y}$, $f_{X,Y}$ represent the characteristic functions of $X$, $Y$, and $(X,Y)$ respectively.

The rest follows from the proof of Proposition~1 in \cite{zhang2015direction}. For any $\bm{\beta}\neq\bC$ that satisfies $\bm{\beta}^T\Sigma_X\bm{\beta}=\bI$, let $\bm{\beta}_1$ be the projection of $\bm{\beta}$ to the subspace spanned by $\bC$ with an inner product induced by $\Sigma_X$ (that is, $\Sigma_X^{0.5}\bm{\beta}_1$ being the projection of $\Sigma_X^{0.5}\bm{\beta}$ to the subspace spanned by $\Sigma_X^{0.5}\bC$ under the Euclidean metric) and $\bm{\beta}_2=\bm{\beta}-\bm{\beta}_1$. Then since $\Sigma_X^{0.5}\bm{\beta}$ and $\Sigma_X^{0.5}\bC$ are both orthogonal subspaces, we have $\|\bC^{\dagger}\bm{\beta}_1\|=\|(\Sigma_X^{0.5}\bC)^{\dagger}(\Sigma_X^{0.5}\bm{\beta}_1)\|\leq \|(\Sigma_X^{0.5}\bC)^{\dagger}(\Sigma_X^{0.5}\bm{\beta})\|\leq 1$, where ${\dagger}$ represents the pseudo inverse. Note that $\bm{\beta}_1$ and $\bC$ have the same column space, so for any $\bz\in\reals^p$, we have
\begin{equation}\label{eq:comparebc}
\|\bm{\beta}_1\bz\|\leq \|\bC\bz\|.
\end{equation}  
Then we proved that $\bC$ is the solution to $\arg\min_{\bC^T\Sigma_X\bC=I_d}\nu^2(\bC^T\bX,\bY)$ asymptotically:
\begin{align*}
&\nu^2(\bm{\beta}^TX,Y,\alpha)=\int|\Expect e^{i\langle t,\bm{\beta}^TX\rangle+i\langle s, Y \rangle}-\Expect e^{i\langle t,\bm{\beta}^TX\rangle}\Expect e^{i\langle s, Y \rangle}|^2\big/(\|t\|^{d+\alpha}\|s\|^{1+\alpha})\di t\di s\\
=&\int|\Expect e^{i\langle t,\bm{\beta}_2^TX\rangle}|^2|\Expect e^{i\langle t,\bm{\beta}_1^TX\rangle+i\langle s, Y \rangle}-\Expect e^{i\langle t,\bm{\beta}_1^TX\rangle}\Expect e^{i\langle s, Y \rangle}|^2\big/(\|t\|^{d+\alpha}\|s\|^{1+\alpha})\di t\di s\\
\leq &\int|\Expect e^{i\langle t,\bm{\beta}_1^TX\rangle+i\langle s, Y \rangle}-\Expect e^{i\langle t,\bm{\beta}_1^TX\rangle}\Expect e^{i\langle s, Y \rangle}|^2\big/(\|t\|^{d+\alpha}\|s\|^{1+\alpha})\di t\di s\\
 = &\nu^2_n(\bm{\beta}_1^TX,Y,\alpha)\leq \nu^2(\bC^TX,Y,\alpha),
\end{align*}
where the last step follows from Equation~\eqref{eq:comparebc}. It is easy to verify that the equality only holds when $\bm{\beta}=\bC Q$ for some rotation matrix $Q$.

It remains to shows that $\nu^2_n(\bC^T \bX, \bY,\alpha)$ is the empirical estimate of the random variable $\nu^2(\bC^T X, Y,\alpha)$, which means that $\nu^2_n(\bC^T\bX,\bY,\alpha)\stackrel{a.s.}{\to}\nu^2(\bC^T\bX,\bY,\alpha)$ (almost sure convergence) as $n \to\infty$. The result holds following the proof of Lemma 2 in the supplementary material of \cite{zhang2015direction}. 
\end{proof}

\subsection{Convergence analysis}

We investigate the convergence property of the proposed algorithm in this section. In fact, the proposed algorithm generates solutions that converge
to a stationary point of $F_{\eta}(\bC)$ as $t\to\infty$. In addition, the algorithm converges to the solution when well-initialized. 

\begin{theorem}
(a) Any accumulation point of the sequence $\left\{\hat{\bC}^{(t)}\right\}_{t\geq 0}$ generated by the proposed algorithm converges is a stationary point of $F_{\eta}(\bC)$ over the Stiefel manifold.

%
(b) If in addition, the global maximizer $\hat{\bC}$ it is the unique stationary point in its neighborhood $\mathcal{N}$, and  $F_{\eta}(\bC)-F_{\eta}(\hat{\bC})\leq -c\|\bC-\hat{\bC}\|_F^2$ for any $\bC$ in $\mathcal{N}$ and some $c>0$. Then when the initialization $\hat{\bC}^{(0)}$ is sufficiently close to $\hat{\bC}$, the sequence $\left\{\hat{\bC}^{(t)}\right\}_{t\geq 0}$ converges to $\hat{\bC}$. 
\end{theorem}
\begin{proof} 
(a) Due to the line search strategy in Algorithm~\ref{alg:matrix-reg-log}, the objective value of the objective function is monotonically nondecreasing and as a result, $\nu^2_n(\hat{\bC}^{(t)\,T}\bX,\bY,\alpha)$ converges. Let $\tilde{\bC}$ be any accumulation point of the sequence $\hat{\bC}^{(t)}$, then $\nabla_{\bC}\nu^2_n(\bC^{T}\bX,\bY,\alpha)|_{\bC=\hat{\bC}^{(t)}}=0$, since otherwise the objective function will continue to increase. 

(b) Since the gradient of $F_{\eta}(\bC)$ is continuous, $\max_{\bC:\|\bC-\hat{\bC}\|_F\leq \epsilon}\|F_{\eta}(\bC)\|$ converges to zero as $\epsilon\rightarrow 0$. As a result, we may choose $\epsilon'>0$ such that for 
\[
\mathcal{N}_{\epsilon'}=\mathcal{N}\cap \{F_{\eta}(\bC)-F_{\eta}(\hat{\bC})>-\epsilon'\},
\]
and any $\hat{\bC}^{(t)}\in \mathcal{N}_{\epsilon'}$, $\|\hat{\bC}^{(t)}-\hat{\bC}\|_F\leq \sqrt{\epsilon'/c}$ and the gradient $F_{\eta}'(\hat{\bC}^{(t)})$ is so small such that the next iteration $\hat{\bC}^{(t+1)}$ remains in $\mathcal{N}$. Since the functional value $F_{\eta}(\hat{\bC}^{(t)})$ is nonincreasing, $\hat{\bC}^{(t+1)}$ lies in $\mathcal{N}_{\epsilon'}$ as well. As $\hat{\bC}$ is the unique stationary point in $\mathcal{N}_{\epsilon'}$, part (a) implies that the algorithm converges to $\hat{\bC}$.


\end{proof}
\section{Numerical Studies}

In this section, we perform a comparative analysis of several algorithms including the proposed robust SDR (rSDR), the SQP algorithm \citep{sheng2013direction},the MMRN algorithm~\citep{wu2021mm} and the HSIC algorithm~\citep{zhang2015direction}.

The problem in Equation~\eqref{eq:dcov} is nonlinear and the proposed algorithm, rSDR, needs an good initialization. The solutions of the sliced inverse regression (SIR, \cite{li1991sliced}) and the directional regression (DR, \cite{li2007directional}) are used in the initialization of \Cref{alg:matrix-reg-log}. Let $\bm\beta_1$ and $\bm{\beta}_2$ be two solutions of SDR obtained by SIR and DR, respectively. We select one of $\bm\beta_1$ and $\bm{\beta}_2$ with larger dCov as our initial value of $\bm\beta$. Let $\hat{\bm{\Sigma}}$ be the sample covariance of $\{\xbm\}_{i=1}^n$. The initial matrix $\bC^{(0)}=\hat{\bm{\Sigma}}^{1/2}_X\bm{\beta}$ is evaluated in \Cref{alg:matrix-reg-log}.

The proposed algorithm has an parameter $\alpha$ which governs robustness to outliers. A smaller $\alpha$ usually enhances the robustness of \Cref{alg:matrix-reg-log}. However, an excessively small $\alpha$ often results in numerous local minimum values for the problem. Therefore, $\alpha$ is tuned through 5-fold cross-validation.
 The value of $\alpha$ is fine-tuned from $\{i/10\}_{i=1}^{9}$ by 5-fold CV. Specifically, we partition the datasets $\{(y_i,\xbm_i)\}_{i=1}^n$ into training and validation sets. For each $\alpha$ value, we apply \Cref{alg:matrix-reg-log} to the training set, yielding a subspace $\bm{\beta}_\alpha$. We then assess the $0.5$-dCov of the validation set. This process is repeated for all 5 folds, and the average $0.5$-dCov is computed. We choose the $\alpha$ value associated with the highest average and execute \Cref{alg:matrix-reg-log} again to derive the estimated subspace. It is important to note that if the dataset is contaminated with outliers, the validation set will also contain outliers. Traditional dCov or covariance calculations may be significantly impacted by these outliers. Therefore, opting for a more robust variance statistic is crucial. In this context, we select the $0.5$-dCov as the measure for the test set.

The SQP algorithm utilizes sequential quadratic programming to solve the dCov-based SDR model (equivalent to Equation~\eqref{eq:dcov} with $\alpha=1$). While the SQP method performs well when the dimension ($p$) and sample size ($n$) are relatively small, it becomes computationally difficult for moderately high-dimensional settings~\citep{wu2021mm}. MMRN was later proposed as an efficient alternative to solve the same model using Riemannian Newton’s method. Both SQP and MMRN correspond to rSDR with $\alpha=1$, but none of them is robust against outliers. The Hilbert-Schmidt Independence Criterion (HSIC) method~\citep{zhang2015direction} addresses the single-index SDR model ($d=1$) by maximizing the HSIC covariance between $\bm{\beta}^T\bX$ and $\bY$.

In the first simulation, we compare rSDR with SQP and MMRN in both robust and non-robust settings. Our results demonstrate that rSDR with a smaller $\alpha$ can effectively estimate the underlying subspace and efficiently solve the SDR model. Additionally, even in the presence of outliers in the response, rSDR can still estimate the subspace accurately, while SQP and MMRN fail to do so.

In the second simulation, we explore the application of rSDR in outlier detection. By reducing the data dimension, we extend a dCor-based outlier detection method~\citep{wang2017outlier} to high-dimensional cases. We compare rSDR with PCA in dimensionality reduction and outlier detection to showcase the applicability of robust SDR in outlier detection.

Furthermore, we present three real data examples: the New Zealand horse mussels, cardiomyopathy microarray data, and auto MPG data. In the New Zealand horse mussels dataset, we reduce the data dimension to 1 and compare rSDR with HSIC. Notably, HSIC is only applicable when $d=1$, so we do not include it in other simulations or real data examples.

\subsection{Simulation Data}

Let $\tilde{\beta}_1=(1,0,0,\cdots,0)^T,\tilde{\beta}_2=(0,1,0,\cdots,0)^T,\tilde{\beta}_3=(1,0.5,1,\cdots,0)^T$ be three $p$-dimensional vectors. We further rotate the vectors $\tilde{\beta}_i$ by a random rotation matrix $R_d\in SO(p)$ (the special orthogonal group of dimension $p$), i.e., $\beta_i=R_d^\top\tilde{\beta}_i$. We consider the following three models
\begin{enumerate}
    \item[(A)] $Y = (\beta_1^TX)^2+(\beta_2^TX)+0.1\epsilon$, 
    \item[(B)] $Y = \sign(2\beta_1^TX+\epsilon_1)\times\log|2\beta_2^TX+4+\epsilon_2|$, 
    \item[(C)] $Y = \exp(\beta_3^TX)\epsilon$,
\end{enumerate}
where $X\in\reals^p$ follows from (1) $\Ncal(0,\Ibf)$ and (2) $U[-2,2]^p$ and $\epsilon,\epsilon_1,\epsilon_2$ are standard normal distributed. We analyze the principal angles between the true subspace $\bm{\beta}$ and the estimated subspace $\hat{\bm\beta}$ obtained using different SDR methods, namely rSDR, MMRN, and SQP. To further investigate the robustness of these methods, we introduce additional noise by adding the response with a value of $50\times \bm{1}^TX$ with a probability of 0.1. We then calculate the principal angles between the true subspace and the estimated subspaces in this robust setting. Both simulation scenarios are conducted for two settings: $(n,p)=(100,6)$ and $(n,p)=(500,20)$. We repeat the simulations 100 times and report the mean and standard deviation of the principal angles for both the non-robust and robust cases in Tables \ref{table:1} and \ref{table:2}, respectively. It is worth noting that the underlying subspace for model (A) and (B) is represented by $\bm{\beta}=[\beta_1,\beta_2]$, resulting in a value of $d=2$. On the other hand, the underlying subspace for model (C) is represented by $\bm{\beta}=\beta_3$, resulting in a value of $d=1$.


From Table~\ref{table:1} we observe that rSDR performs better than MMRN and SQP in model (A) and (B) even without outliers. MMRN converges faster than SDR in model (A) and (B). When $(n,p)=(500,20)$, rSDR and MMRN are faster than SQP. Table~\ref{table:2} reports the principal angles and execution time of the three estimators in the scenario where the outliers present. Table~\ref{table:2} shows that the principal angles between the true subspace and the estimated subspace produced by rSDR are smaller than MMRN and SQP which implies that rSDR is more robust. Moreover, rSDR converges faster than MMRN and SQP in most settings; particularly in model (C). 

\begin{table}
\centering
\tbl{The mean and standard deviation (in parentheses) of the principal angles and the running times (seconds) over 100 repetitions of SQP, MMRN and rSDR in nonrobust settings.}{%
\begin{tabular}{lcccccccc} \toprule
\multirow{2}{*}{$(n,p)$} & \multirow{2}{*}{Model} & \multicolumn{2}{c}{SQP} & \multicolumn{2}{c}{MMRN} & \multicolumn{2}{c}{rSDR} \\ \cmidrule(lr){3-4} \cmidrule(lr){5-6} \cmidrule(lr){7-8}
& & \multicolumn{1}{c}{Angle} & \multicolumn{1}{c}{Time(s)} & \multicolumn{1}{c}{Angle} & \multicolumn{1}{c}{Time(s)} & \multicolumn{1}{c}{Angle} & \multicolumn{1}{c}{Time(s)} \\ \midrule
\multirow{6}{*}{(100,6)}  & A(1)                   & 0.27(0.09) & 0.16(0.12) & 0.27(0.09)  & 0.19(0.09) & 0.27(0.09)  & 0.17(0.13) \\
                          & A(2)                   & 0.25(0.08) & 0.13(0.13) & 0.25(0.08)  & 0.15(0.12) & 0.25(0.08)  & 0.18(0.13) \\
                          & B(1)                   & 0.28(0.09) & 0.10(0.02) & 0.28(0.09)  & 0.19(0.12) & 0.28(0.09)  & 0.20(0.14) \\
                          & B(2)                   & 0.22(0.08) & 0.11(0.04) & 0.22(0.08)  & 0.32(0.59) & 0.21(0.08)  & 0.23(0.18) \\
                          & C(1)                   & 0.20(0.07) & 0.24(0.32) & 0.20(0.07)  & 0.25(0.11) & 0.19(0.06)  & 0.08(0.05) \\
                          & C(2)                   & 0.32(0.12) & 0.14(0.24) & 0.31(0.12)  & 0.38(0.17) & 0.32(0.12)  & 0.08(0.05) \\
\midrule
\multirow{6}{*}{(500,20)} & A(1)                   & 0.24(0.04) & 2.98(0.56) & 0.24(0.04)  & 0.90(0.14) & 0.24(0.04)  & 1.41(0.77) \\
                          & A(2)                   & 0.23(0.04) & 3.33(3.11) & 0.23(0.04)  & 0.90(1.36) & 0.23(0.04)  & 1.65(0.80) \\
                          & B(1)                   & 0.24(0.04) & 3.17(0.63) & 0.24(0.04)  & 0.90(0.14) & 0.24(0.04)  & 1.54(0.91) \\
                          & B(2)                   & 0.19(0.03) & 4.55(1.29) & 0.19(0.03)  & 0.81(0.13) & 0.18(0.03)  & 1.56(0.79) \\
                          & C(1)                   & 0.16(0.03) & 2.54(0.22) & 0.16(0.03)  & 1.52(0.26) & 0.17(0.03)  & 0.66(0.42) \\
                          & C(2)                   & 0.25(0.04) & 3.27(0.71) & 0.25(0.04)  & 4.04(1.23) & 0.28(0.05)  & 0.72(0.49) \\
\bottomrule
\end{tabular}}
\label{table:1}
\end{table}

\begin{table}
\centering
\tbl{The mean and standard deviation (in parentheses) of the principal angle, and the running time (seconds) over 100 repetitions of SQP, MMRN and rSDR in robust settings.}{%
\begin{tabular}{lcccccccc} \toprule
\multirow{2}{*}{$(n,p)$} & \multirow{2}{*}{Model} & \multicolumn{2}{c}{SQP} & \multicolumn{2}{c}{MMRN} & \multicolumn{2}{c}{rSDR} \\
\cmidrule(lr){3-4} \cmidrule(lr){5-6} \cmidrule(lr){7-8}
& & \multicolumn{1}{c}{Angle} & \multicolumn{1}{c}{Time(s)} & \multicolumn{1}{c}{Angle} & \multicolumn{1}{c}{Time(s)} & \multicolumn{1}{c}{Angle} & \multicolumn{1}{c}{Time(s)} \\ \midrule
\multirow{6}{*}{(100,6)}  & A(1)                   & 0.51(0.28) & 0.20(0.21)   & 0.49(0.27) & 0.28(0.20)  & 0.32(0.12)  & 0.19(0.13) \\
                          & A(2)                   & 0.45(0.28) & 0.13(0.12)   & 0.44(0.27) & 0.26(0.19)  & 0.27(0.11)  & 0.17(0.12) \\
                          & B(1)                   & 0.52(0.26) & 0.13(0.08)   & 0.51(0.26) & 0.32(0.26)  & 0.33(0.12)  & 0.18(0.13) \\
                          & B(2)                   & 0.42(0.22) & 0.12(0.07)   & 0.42(0.22) & 0.25(0.22)  & 0.24(0.08)  & 0.21(0.17) \\
                          & C(1)                   & 0.39(0.24) & 0.19(0.21)   & 0.38(0.23) & 0.32(0.16)  & 0.26(0.10)  & 0.08(0.06) \\
                          & C(2)                   & 0.47(0.23) & 0.17(0.27)   & 0.46(0.22) & 0.51(0.25)  & 0.40(0.16)  & 0.09(0.06) \\
\midrule
\multirow{6}{*}{(500,20)} & A(1)                   & 0.82(0.30) & 4.12(1.29)   & 0.82(0.30) & 2.92(1.47)  & 0.25(0.04)  & 1.47(0.87) \\
                          & A(2)                   & 0.93(0.42) & 16.26(42.21) & 0.92(0.42) & 6.85(13.70) & 0.24(0.04)  & 1.43(0.85) \\
                          & B(1)                   & 0.91(0.35) & 4.20(1.48)   & 0.90(0.35) & 3.36(2.23)  & 0.26(0.04)  & 1.74(0.92) \\
                          & B(2)                   & 0.60(0.36) & 5.32(2.26)   & 0.60(0.36) & 2.84(3.54)  & 0.19(0.03)  & 1.64(0.99) \\
                          & C(1)                   & 0.35(0.14) & 3.02(0.34)   & 0.35(0.14) & 4.41(2.15)  & 0.22(0.04)  & 0.63(0.46) \\
                          & C(2)                   & 0.89(0.28) & 4.53(1.01)   & 0.85(0.29) & 13.39(6.27) & 0.35(0.07)  & 0.88(0.61) \\
\bottomrule
\end{tabular}}
\label{table:2}
\end{table}

\subsection{Outlier Detection Simulation Studies}\label{sec:out-det}


Our proposed SDR method can be effectively utilized for outlier detection. \cite{wang2017outlier} introduced a novel outlier detection measure based on the distance correlation (dCor) given by
\begin{align}\label{eq:Di}
\Dcal_i(\Xb,\bY) = \frac{1}{p}\sum_{k=1}^p\left(\dCor(\Xb_k,\bY)-\dCor(\Xb^{(i)}_k,\bY^{(i)}) \right)^2,
\end{align}
where $\dCor(\Xb_k,\bY)$ represents the dCor between the $k$-th predictor and the response $\bY$. The dCor between $\bX$ and $\bY$ is defined as
\begin{align*}
\dCor^2(\bX,\bY) = \frac{\dCov^2(\bX,\bY)}{\sqrt{\dCov^2(\bX,\bX)\dCov^2(\bY,\bY)}}.
\end{align*}
It is evident that if the $i$-th data point $(\Xb^{(i)}_k,\bY^{(i)})$ exhibits a high value of the measure $\hat{\Dcal}_i$, it is more likely to be an outlier observation. The method employs a bootstrap procedure to determine the threshold $\hat{F}_{\gamma}$. At a given significance level $\gamma$, the $i$-th observation is identified as an outlier if $\hat{\Dcal}_i>\hat{F}_{\gamma}$, where $\hat{F}_{\gamma}$ represents the upper $\gamma$-th quantile of the cumulative distribution function of $\Dcal_i$ under the null hypothesis. Specifically, a bootstrap sample $\Dcal^{[b]}_i$ is formed by drawing with replacement from ${1,\cdots,n}$, denoted as ${i^{[b]}_{(1)},\cdots,i^{[b]}_{(n)} }$, and an estimator $\hat{\Dcal}^{[b]}_i$ is computed for each sample. The threshold $\hat{F}_{\gamma}$ is determined by calculating the upper $\gamma$-th quantile of the cumulative distribution function of $\hat{\Dcal}^{[b]}_i$.

The algorithm proposed by \cite{wang2017outlier}, which is based on the outlier detection measure defined in Equation~\eqref{eq:Di}, involves calculating the covariance distance between $\bX$ and $\bY$ in each dimension and with the removal of each sample. As a result, its computational complexity is $O(pn^3)$, where the computation of dCov requires pairwise distance calculations between the columns of $\bX$ and $\bY$. A natural approach to enhance their method is to reduce the dimensionality of the dataset $\bX$. Their method can be naturally extended to detect outlier locations by computing
\begin{align}\label{eq:Di-pca}
\Dcal_i(\Bar{\Xb},\bY) = \frac{1}{p}\sum_{k=1}^p\left(\dCor(\Bar{\Xb}_k,\bY)-\dCor(\Bar{\Xb}^{(i)}_k,\bY^{(i)}) \right)^2,
\end{align}
where $\Bar{\Xb}\in\reals^{d\times n}$ is the $d$-dimensional data obtained by dimension reduction. Nevertheless, the conventional approach to dimension reduction is unsuitable in the presence of outliers. Therefore, we employ the robust SDR as a means to both reduce the data's dimensionality and identify outlier positions. For the sake of comparison, we also implement principal component analysis (PCA) \citep{wold1987principal} for dimension reduction.

We consider an autoregressive correlation structure with $\bm{\Sigma}=(\rho_{j,k})_{p\times p}=0.5^{|j-k|}$ and generate the data as follows: $X_i$ follows a multivariate normal distribution $\Ncal(0,\bm{\Sigma})$, and the linear model is defined as $Y_i=X_i\beta+\epsilon_i$, where $\beta=(1,1,1,1,1,0,\cdots,0)^T$ and $\epsilon_i\sim\Ncal(0,1)$. We have a total of $n=100$ samples, and among them there are $10$ outliers. The outliers are generated using $\kappa_i=X_i\gamma$, where $\gamma=(0,0,0,0,0,1,1,\cdots)$. We did four sets of simulations for various values of $p=200,400,800,1000$. To test the hypothesis of whether the $i$-th observation is influential or not, we employ a bootstrap procedure and utilize a threshold rule to determine whether an individual is an outlier. We evaluate the performance of this outlier identification procedure by comparing the receiver operating characteristic (ROC) curves. 

The ROC curves are depicted in Figure~\ref{fig:roc-out-det}. In the figure, the curve labeled as `PCA-2' represents the ROC curve generated by $\bar{\bX}^{PCA}$ with a dimensionality of $d=2$, while the curve labeled as `rSDR-0.2-2' corresponds to the curve produced by $\bar{\bX}^{DR}$ with $\alpha=0.2$ and $d=2$. Similarly, the remaining labels follow similar settings. It can be observed that the curves generated by rSDR with $d=3$ consistently surpass those produced by rSDR with $d=2$, and both outperform the curves generated by PCA. This suggests that the proposed rSDR method effectively captures the underlying structure of the data, and the resulting transformed data $\bar{\bX}^{DR}$ can be utilized for outlier detection. Notably, despite the true subspace being two-dimensional, $\bar{\bX}^{DR}$ with $d=3$ outperforms its two-dimensional counterpart. We speculate that the higher dimensionality preserves more information due to the presence of outliers.

\begin{figure}[ht]
\centering
\subfloat[$p=200$]{%
\resizebox*{6cm}{!}{\includegraphics{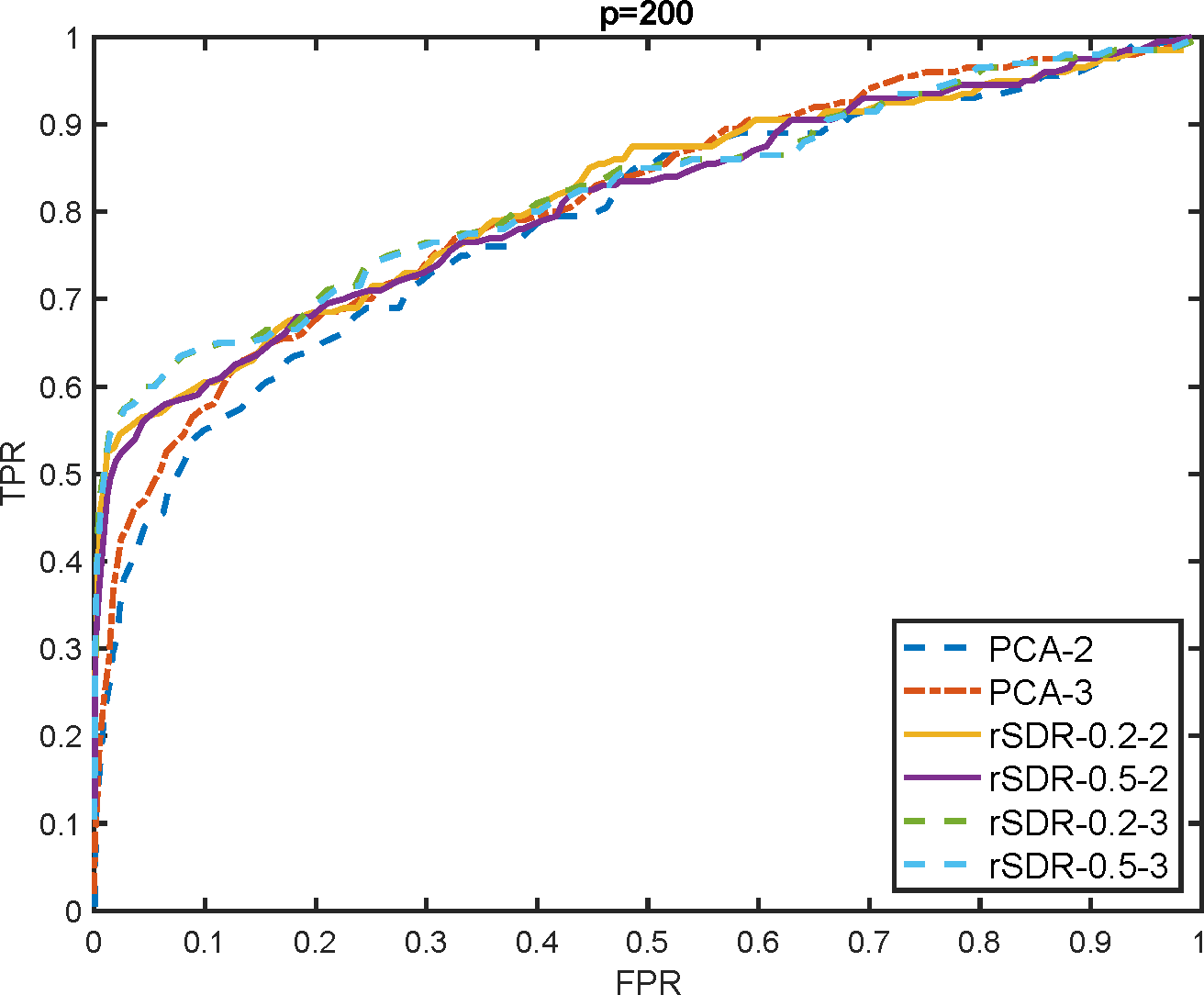}}}\hspace{5pt}
\subfloat[$p=400$]{%
\resizebox*{6cm}{!}{\includegraphics{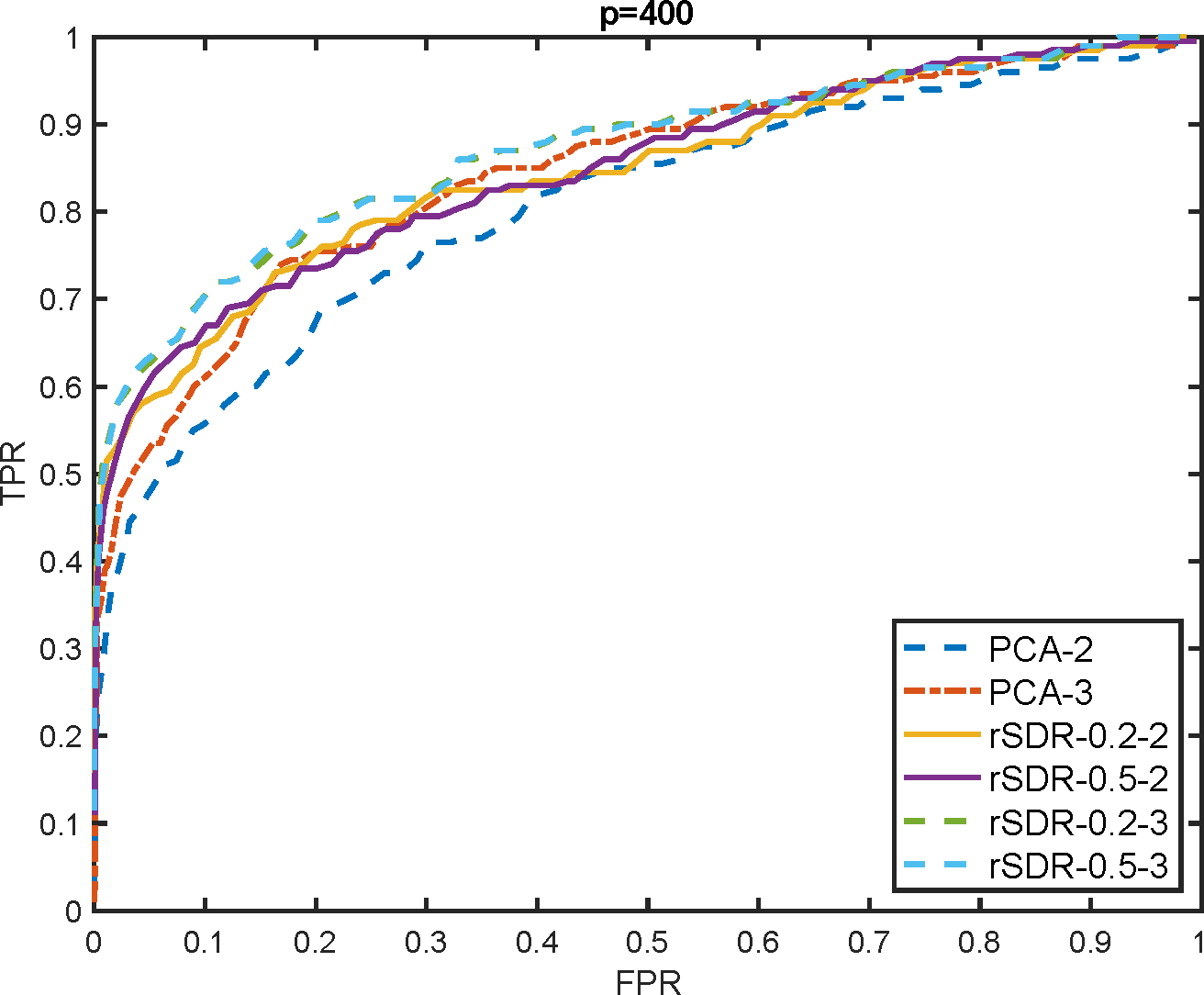}}}\hspace{5pt}
\subfloat[$p=800$]{%
\resizebox*{6cm}{!}{\includegraphics{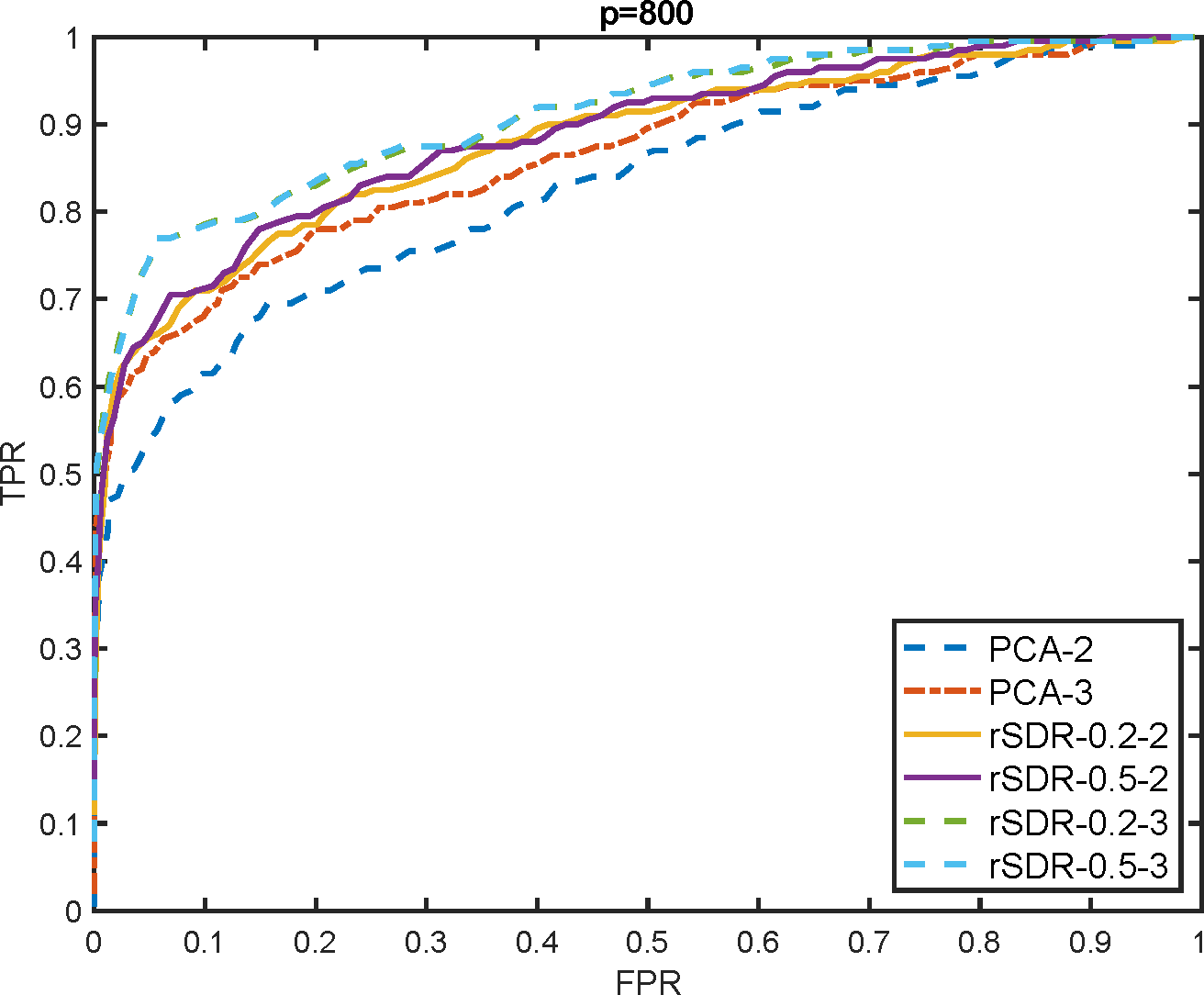}}}\hspace{5pt}
\subfloat[$p=1000$]{%
\resizebox*{6cm}{!}{\includegraphics{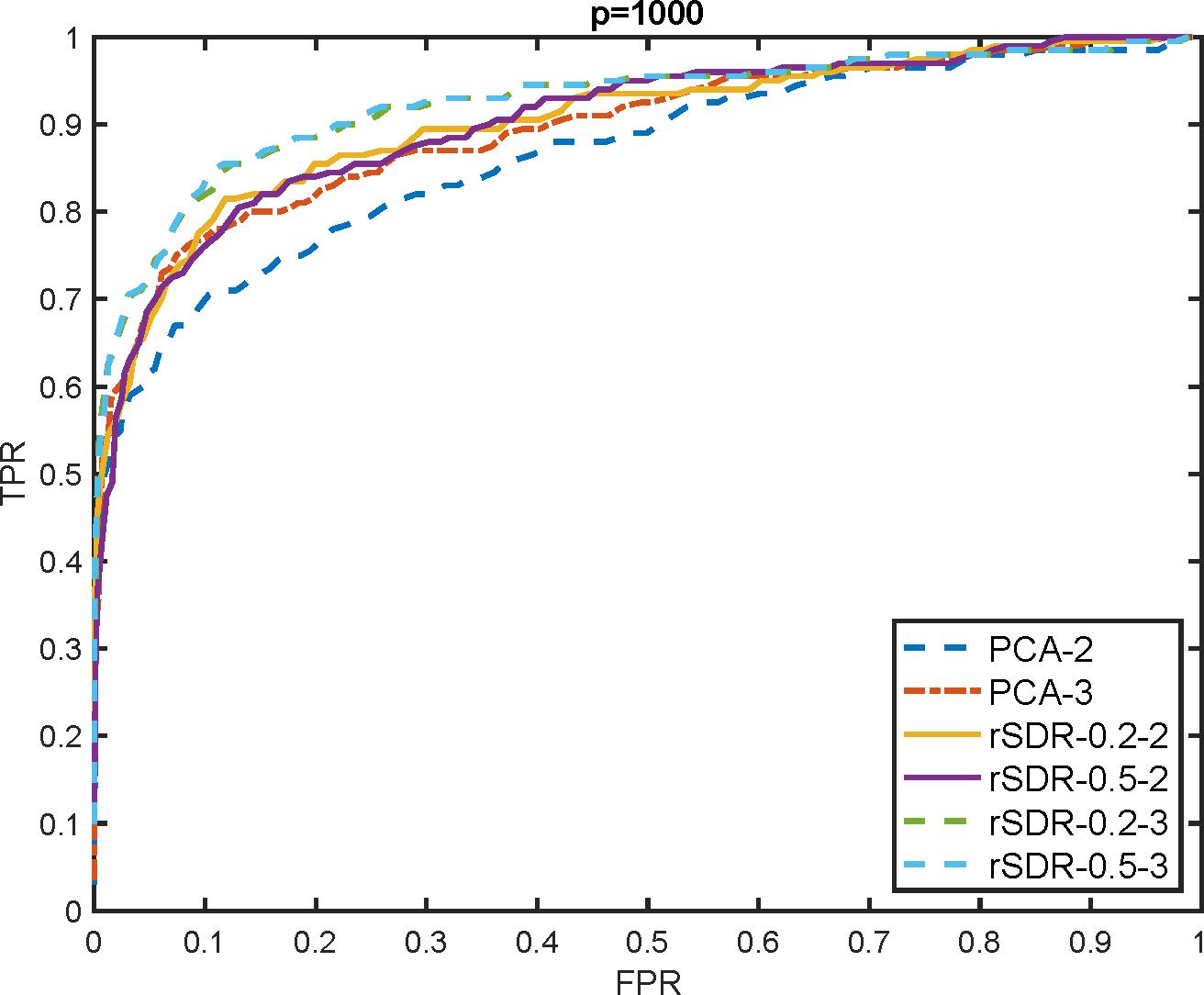}}}\hspace{5pt}
\caption{ROC curves of outlier detection. The proposed robust SDR method with $\alpha=0.5$ with projection dimension $3$ has the highest ROC in these four simulation sample size settings.}
\label{fig:roc-out-det}
\end{figure}

\subsection{Real Data Example: New Zealand Horse Mussels}

A sample of 201 horse mussels (Modiolus modiolus) was collected at 5 sites in the Marlborough Sounds at the Northeast of New Zealand's South Island and this dataset was discussed by \cite{cook2009regression}. The response variable is muscle mass $M$, the edible portion of the mussel, in grams. The quantitative predictors are all related to characteristics of the mussel shells: shell width $W$ (in mm), shell height $H$ (in mm), shell length $H$ (in mm) and shell mass $S$ (in grams). 

To process the data, a nonlinear transformation of the predictors was recommended by \cite{cook2009regression} as $X = (L,W^{0.36},S^{0.11})$. Each column of the data $X$ is further standardized by $\Tilde{X} = (\frac{L-\hat{\mu}_L}{\hat{\sigma}(L)},\frac{W^{0.36}-\hat{\mu}_{W^{0.36}}}{\hat{\sigma}(W^{0.36})},\frac{S^{0.11}-\hat{\mu}_{S^{0.11}}}{\hat{\sigma}(S^{0.11})})$ where $\hat{\mu}_{\cdot}$ is the sample mean and $\hat{\sigma}(\cdot)$ is the sample standard deviation, since $L$ is on a larger scale than the other predictors. Consequently, the predictors will have mean $0$ and variance $1$. The rSDR model with $d=1$ would be appropriate to model this dataset, as shown in  Figure~\ref{fig:new-zealand-mussel}, where we fit two second-degree polynomial regression models of the single index $\hat{\beta}^T\Tilde{X}$ by rSDR with $\alpha=0.2$ and $\alpha=1$. We compare our method rSDP with $\alpha=0.2$ and $\alpha=1$, SQP and the Hilbert-Schmidt Independence Criterion (HSIC) method, proposed by \cite{zhang2015direction} for solving the special case of the SDR model, namely $d=1$. Table~\ref{table:mussel} provides the estimated bases $\beta$ from these four methods. The estimates of SDR with $\alpha=1$ and SQP are similar, and this result is expected since SQP and rSDR with $\alpha=1$ solve the same model with different algorithms. The estimated $\hat{\beta}$ from all four methods indicate that the standardized shell mass predictor, $\tilde{X_3}$, is more significant than the other two predictors while the rSDR with $\alpha=0.2$ produces a smaller value in the coefficient of $\tilde{X_3}$. However, rSDR with $\alpha=0.2$ produces a model with a slightly larger R-squared value than the other methods, which implies a better fit of the dataset. 

\begin{figure}[ht]
\centering
\subfloat[rSDR-0.2]{%
\resizebox*{6cm}{!}{\includegraphics{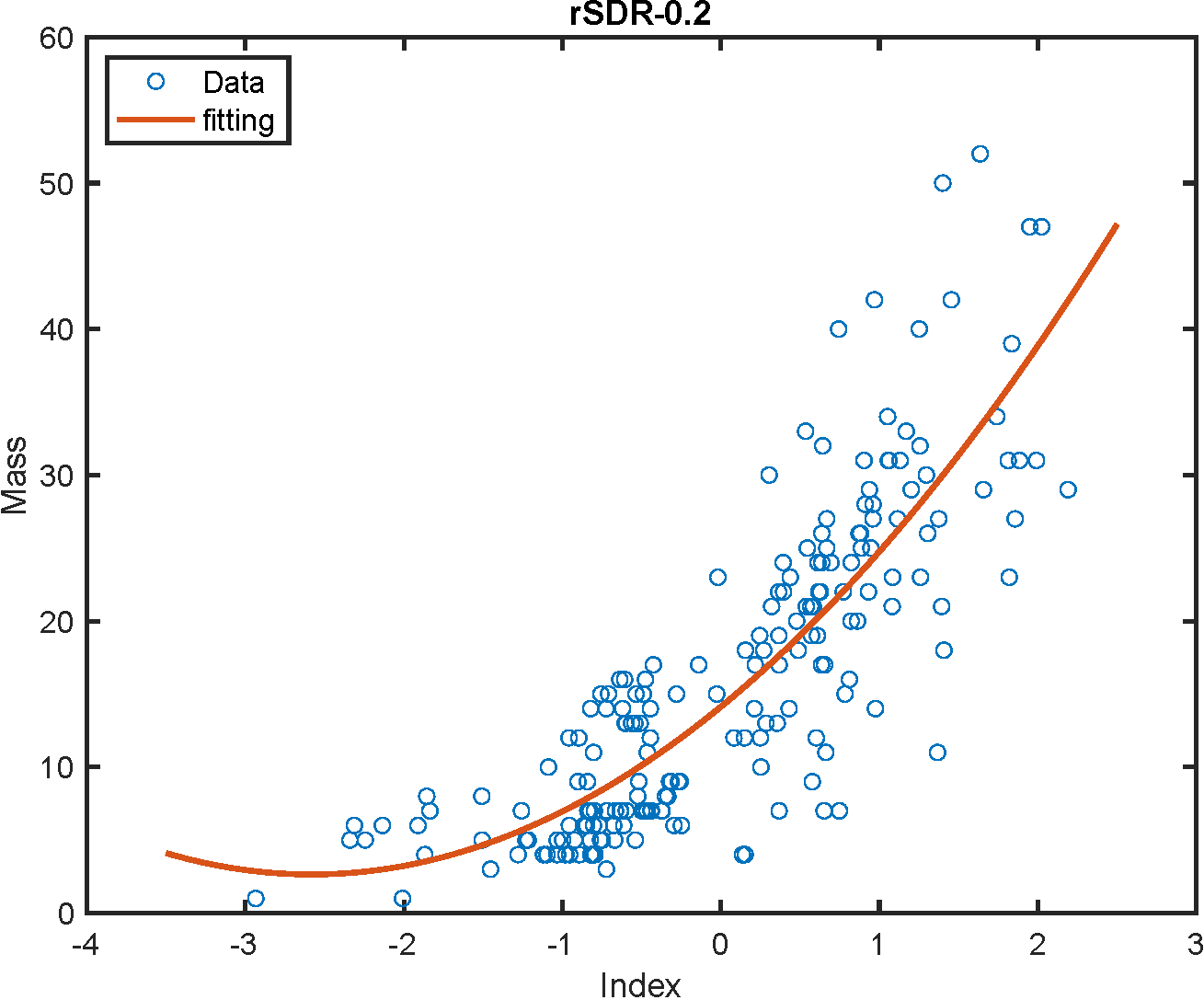}}}\hspace{5pt}
\subfloat[rSDR-1]{%
\resizebox*{6cm}{!}{\includegraphics{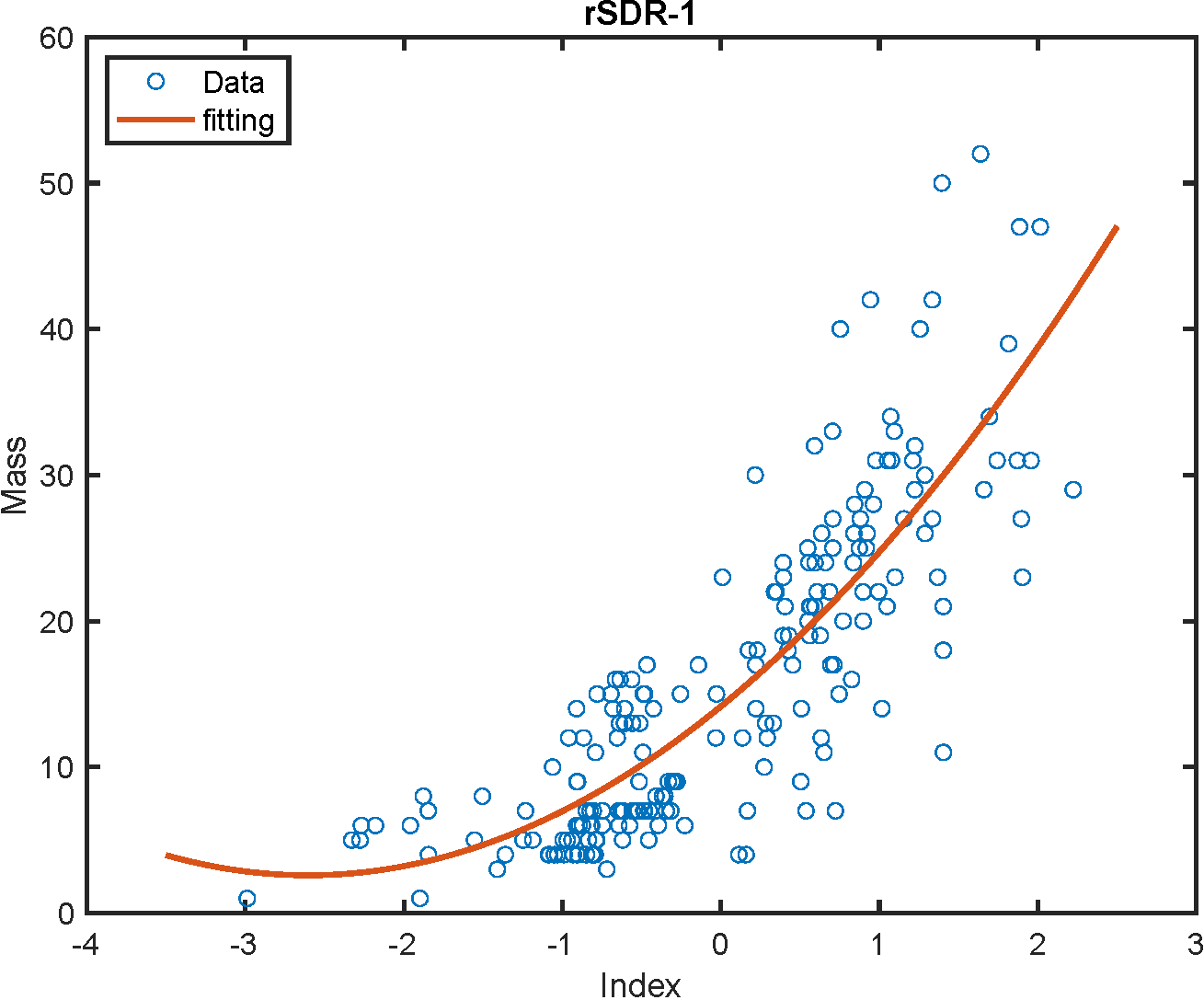}}}\hspace{5pt}
\caption{The second-degree polynomial fitting of the single-index model in the New Zealand Horse Mussels data using rSDR ($\alpha=0.2$) (left) and rSDR ($\alpha=1$) (right).}
\label{fig:new-zealand-mussel}
\end{figure}

\begin{table}[t!]
\centering
\tbl{Estimated bases $\hat{\beta}=[\hat{\beta}_1,\hat{\beta}_2,\hat{\beta}_3]^T\in\Rbb^3$ of the central subspace in the New Zealand Horse Mussels data from various methods and their adjusted R-squared values.}{%
\begin{tabular}{ccccc} \toprule
Method             & rSDR ($\alpha=0.2$) & rSDR ($\alpha=1$) & HSIC    & SQP     \\ \midrule
$\hat{\beta}_1$ & 0.2871              & 0.1832            & 0.1897  & 0.1831  \\ 
$\hat{\beta}_2$ & 0.0872              & -0.0270           & -0.0604 & -0.0269 \\ 
$\hat{\beta}_3$ & 0.6391              & 0.8510            & 0.9800  & 0.8509  \\
Adjusted R-squared &0.7026& 0.6979& 0.6962& 0.6979\\ \bottomrule
\end{tabular}}
\label{table:mussel}
\end{table}

\subsection{Real Data Example: Cardiomyopathy Microarray Data}

The cardiomyopathy microarray dataset consists of 30 samples and 6319 predictors, originally used by \cite{segal2003regression} to evaluate regression-based approaches for microarray analysis. The focus of many researchers, \cite{zou2008regularized} and \cite{li2012feature}, has been to investigate the relationship between the overexpression of a G protein-coupled receptor (Ro1) in mice and the 6319 associated genes. However, due to the high dimensionality of the data compared to the limited number of samples, the sample covariance matrix is not invertible. To address this issue, several methods have been proposed, including SIS (Sure Independence Screening, \cite{fan2008sure}), DCSIS (Distance Correlation SIS, \cite{li2012feature}), BCSIS (Ball Correlation SIS, \cite{pan2019generic}), and SDRLS (Sequential Dimension Reduction for Large $p$ Small $n$ problem, \cite{yin2015sequential}). While SIS, DCSIS, and BCSIS are feature screening methods that rank predictors based on a utility measure, they may not be robust against outliers. Specifically, a set of predictors $\mathcal{A}=\{i\mid U(X_i,Y)>\tau,i=1,\cdots,n\}$ is determined for some threshold $\tau$ and pre-selected utility measure $U$. SDRLS takes a different approach. SDRLS partitions the data set into $X=[X_1,X_2]$ with $\dim(X_1)<n$ and applies the SDR model on $(X_1,[X_2,Y])$ to obtain $R(X_1)$. The dimension of $R(X_1)$ is chosen some integer that is smaller than $\dim(X_1)<n$ and thus a new predictor $[R(X_1),R_2]$ is obtained with a smaller dimension. SDRLS iteratively repeats this process to achieve a dimension smaller than the number of samples.

In this experiment, we utilized the SDRLS method to reduce the dimensionality of the cardiomyopathy microarray data and assess the rSDR against heavy-tailed predictors. The final dimension of the dataset was reduced to $p=19$, while the dimension of the central subspace was set to $d=2$. The central subspace is denoted as $\bm{\beta}=[\beta_1,\beta_2]$. Indexes derived from this reduction were obtained by projecting the processed cardiomyopathy microarray dataset $\bX$ onto the subspaces: $\bZ_1=\beta_1^T\bX$ and $\bZ_2=\beta_2^T\bX$. We performed linear and nonlinear regression to model the response variable ``Ro1'' using predictors $\bZ_1$ and $\bZ_2$. In the nonlinear model, we introduced squared terms ($\bZ_1^2$, $\bZ_2^2$) and an interaction term ($\bZ_1\times \bZ_2$) in addition to the linear model. The regression results are presented in Table~\ref{tab:CM}. The findings demonstrate that our proposed method, rSDR, with a smaller value of $\alpha$, outperforms the non-robust version ($\alpha=1$) in both linear and nonlinear models.

\begin{table}[t!]
\centering
\renewcommand{\arraystretch}{1.2}
\tbl{Adjusted R-squared and F-value of models from SQP, rSDR in Cardiomyopathy Microarray dataset.}{%
\begin{tabular}{ccccc} \toprule
Adjusted R-squared    & rSDR ($\alpha=0.2$) & rSDR ($\alpha=0.5$) & rSDR ($\alpha=1$) \\ \midrule
Linear              & 0.826               & 0.817               & 0.804             \\
Nonlinear          & 0.882               & 0.871               & 0.867             \\ \midrule
F-value              & rSDR ($\alpha=0.2$) & rSDR ($\alpha=0.5$) & rSDR ($\alpha=1$) \\ \midrule
Linear           & 70.1                & 65.8                & 60.6              \\
Nonlinear       & 44.4                & 40.3                & 38.9              \\ \bottomrule     
\end{tabular}}
\label{tab:CM}
\end{table}

\subsection{Real Data Example: Auto MPG data}
We also employ the auto fuel economy data to illustrate the advantage of our rSDR method. The dataset contains city-cycle fuel consumption in miles per gallon (MPG) and 7 predictors: cylinders, displacement, horsepower, weight, acceleration, model year and origin. As suggested in \cite{sheng2016sufficient}, we avoid using ``origin'', because it correlates with ``cylinders'' closely. Missing values are deleted, and 392 observations are left for study. In order to investigate the city-cycle fuel consumption in miles per gallon, we assume that this data set fits a sufficient dimension reduction model. As shown in Figure~\ref{fig:mpg-boxplot}, there exist outliers in ``horsepower'' and ``acceleration''. ``cylinders'' and ``displacement'' that are not normally distributed. Therefore, rSDR is appropriate for this data set. 

Following the suggestion of \cite{sheng2016sufficient}, we use the dimension $d=2$ of the central subspace. Let the subspace be $\bm{\beta}=[\beta_1,\beta_2]$ and the auto MPG data be denoted as $\bX$ where each column is centered and scaled to make the variance as $1$. The following procedures are similar to those done in the cardiomyopathy microarray data. The indexes are derived by rSDR, after the linear and nonlinear regression models are constructed to measure the goodness of fit of the two pair of indexes to ``mpg''. In Table~\ref{tab:MPG}, the adjusted R-squared and F-value of the linear model produced by rSDR with $\alpha=0.2$ are larger than other non-robust models however it does not show much superior in the nonlinear regression model (Figure~\ref{fig:mpg-plots}). 

\begin{figure}
    \centering
    \includegraphics[width=0.6\textwidth]{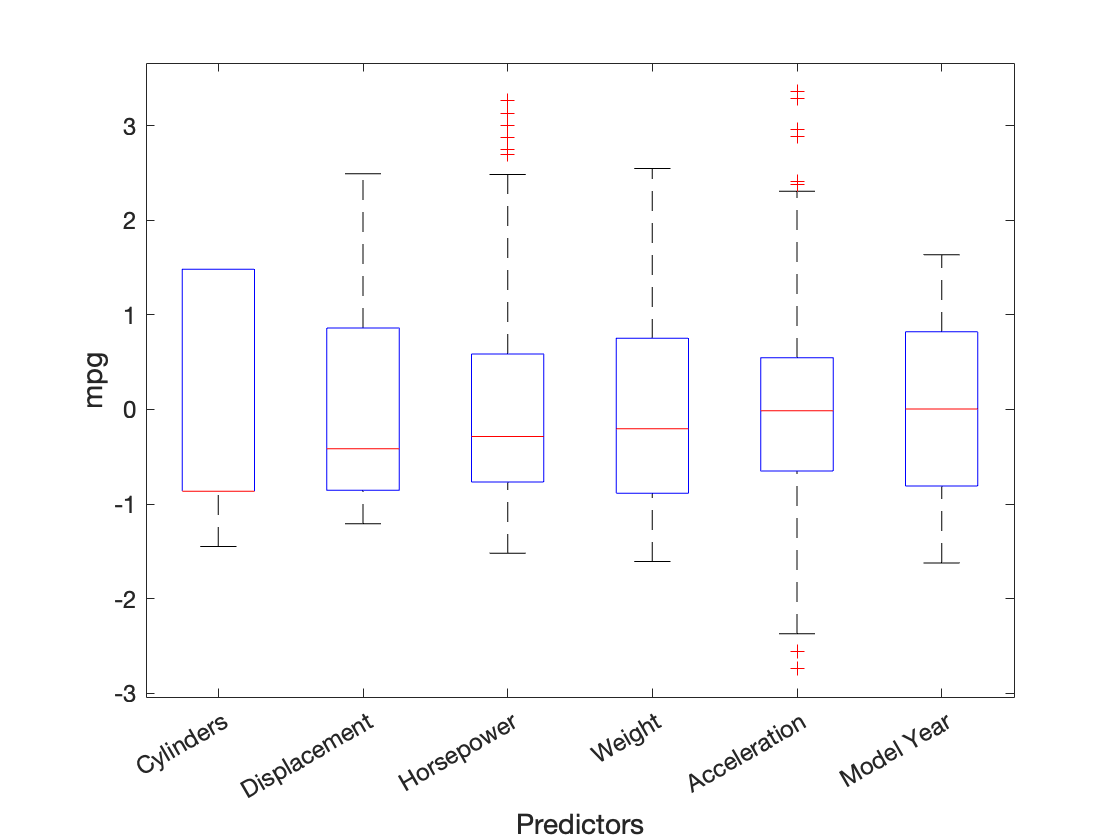}
    \caption{Boxplots of predictors in the auto mpg data set. Some predictors such as ``horsepower'' and ``acceleration'' have outlying observations.}
    \label{fig:mpg-boxplot}
\end{figure}

\begin{figure}[ht]
\centering
\subfloat[SQP-index-1]{%
\resizebox*{6cm}{!}{\includegraphics{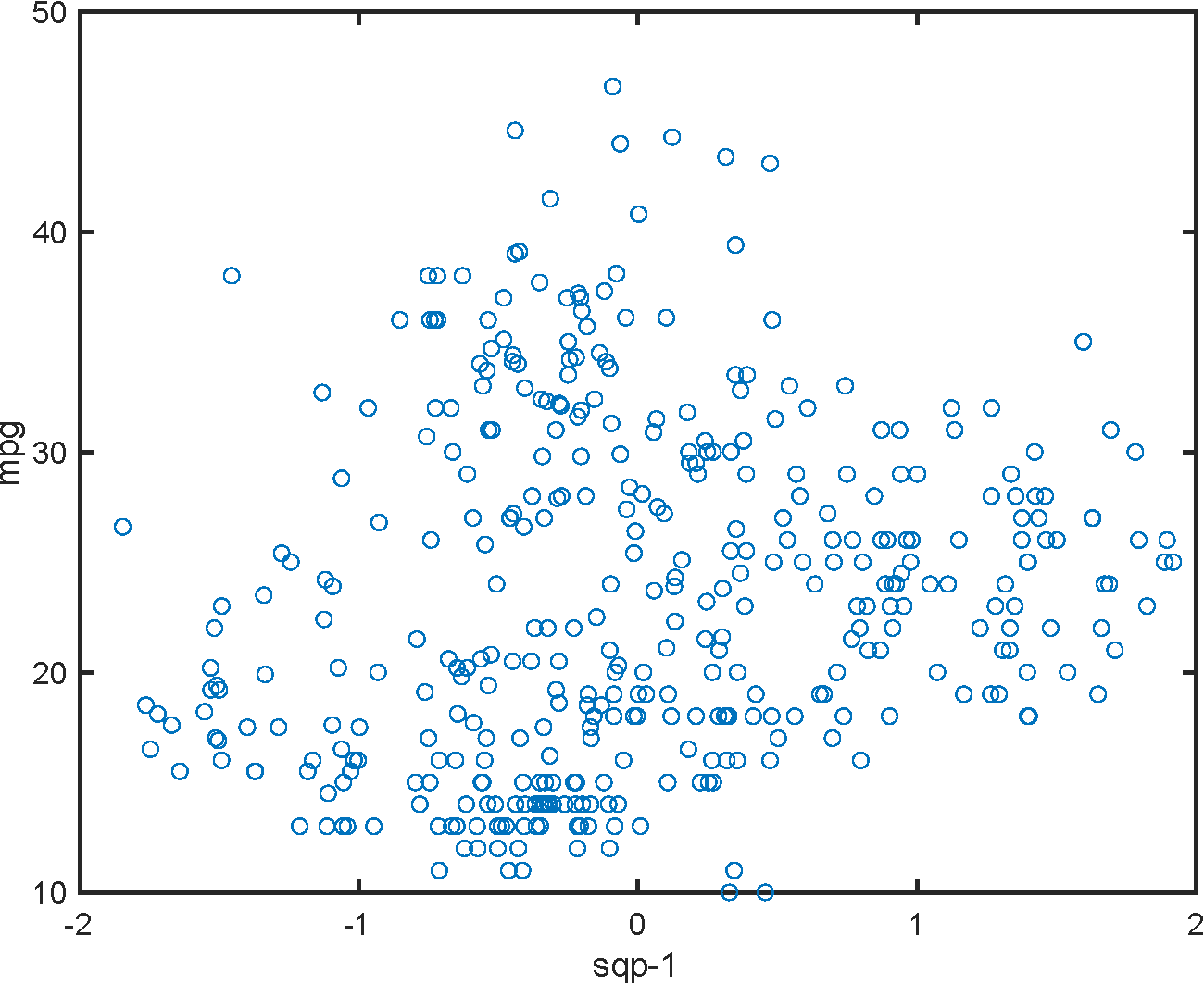}}}\hspace{5pt}
\subfloat[SQP-index-2]{%
\resizebox*{6cm}{!}{\includegraphics{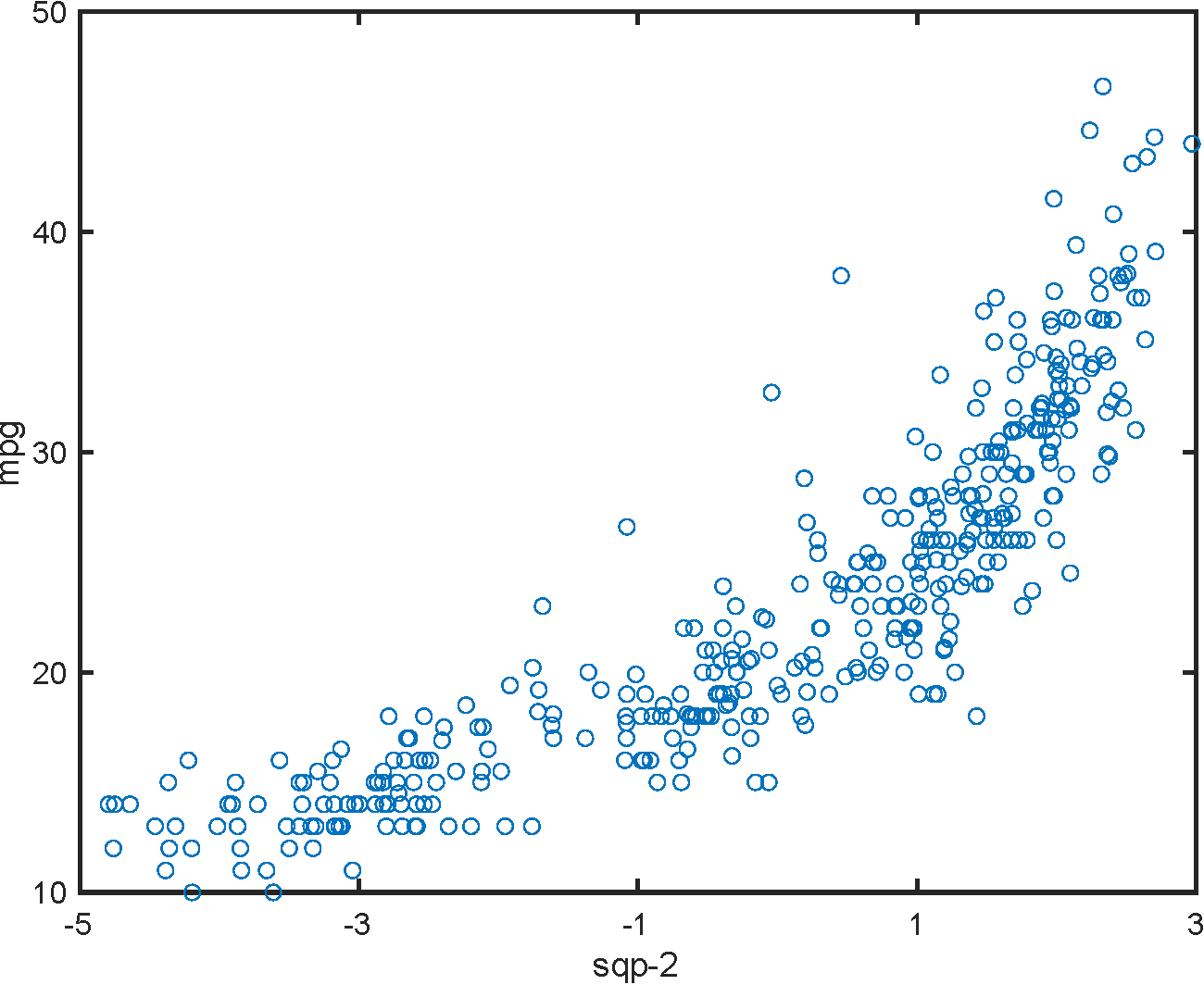}}}\hspace{5pt}
\subfloat[rSDR-index-1]{%
\resizebox*{6cm}{!}{\includegraphics{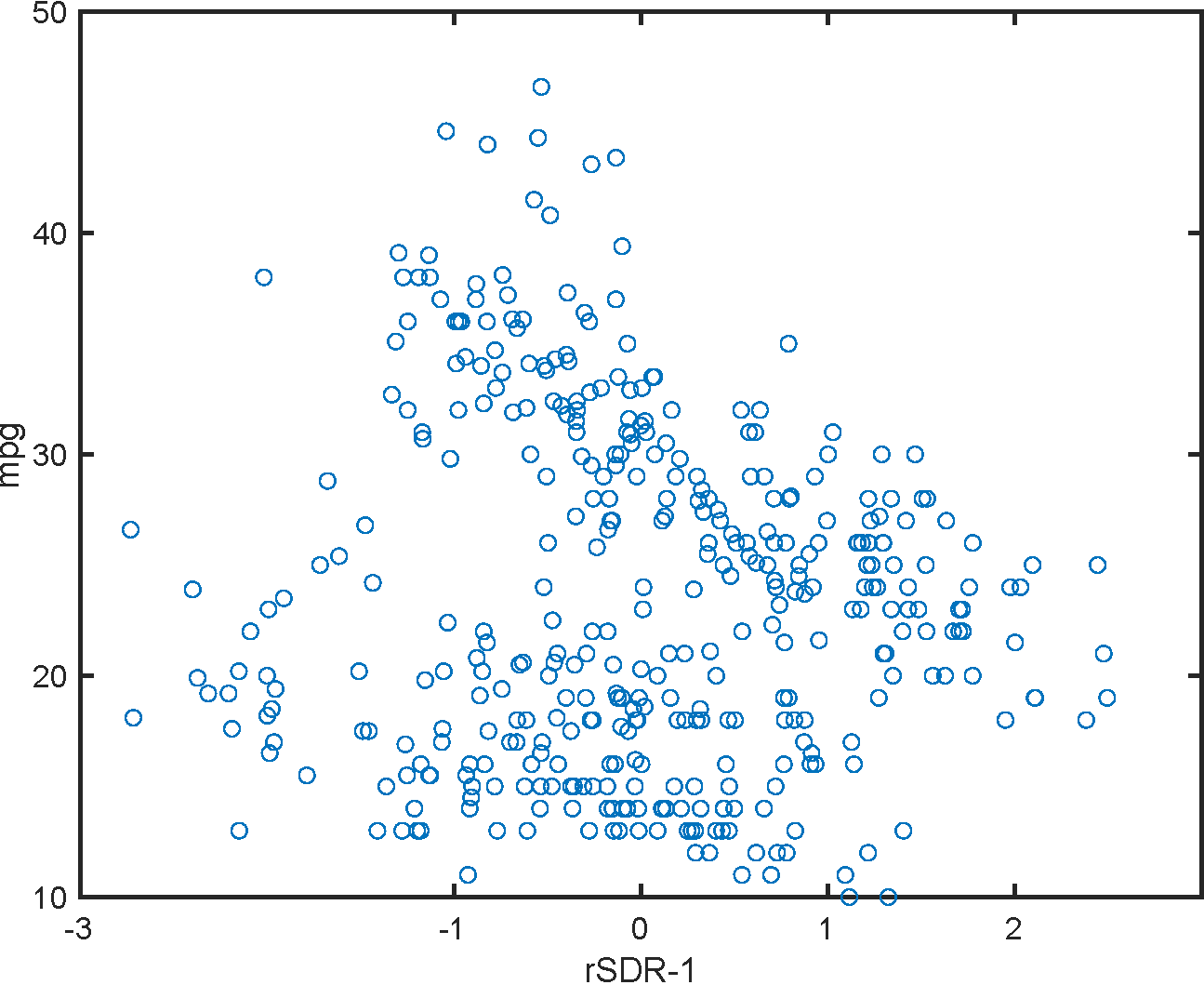}}}\hspace{5pt}
\subfloat[rSDR-index-2]{%
\resizebox*{6cm}{!}{\includegraphics{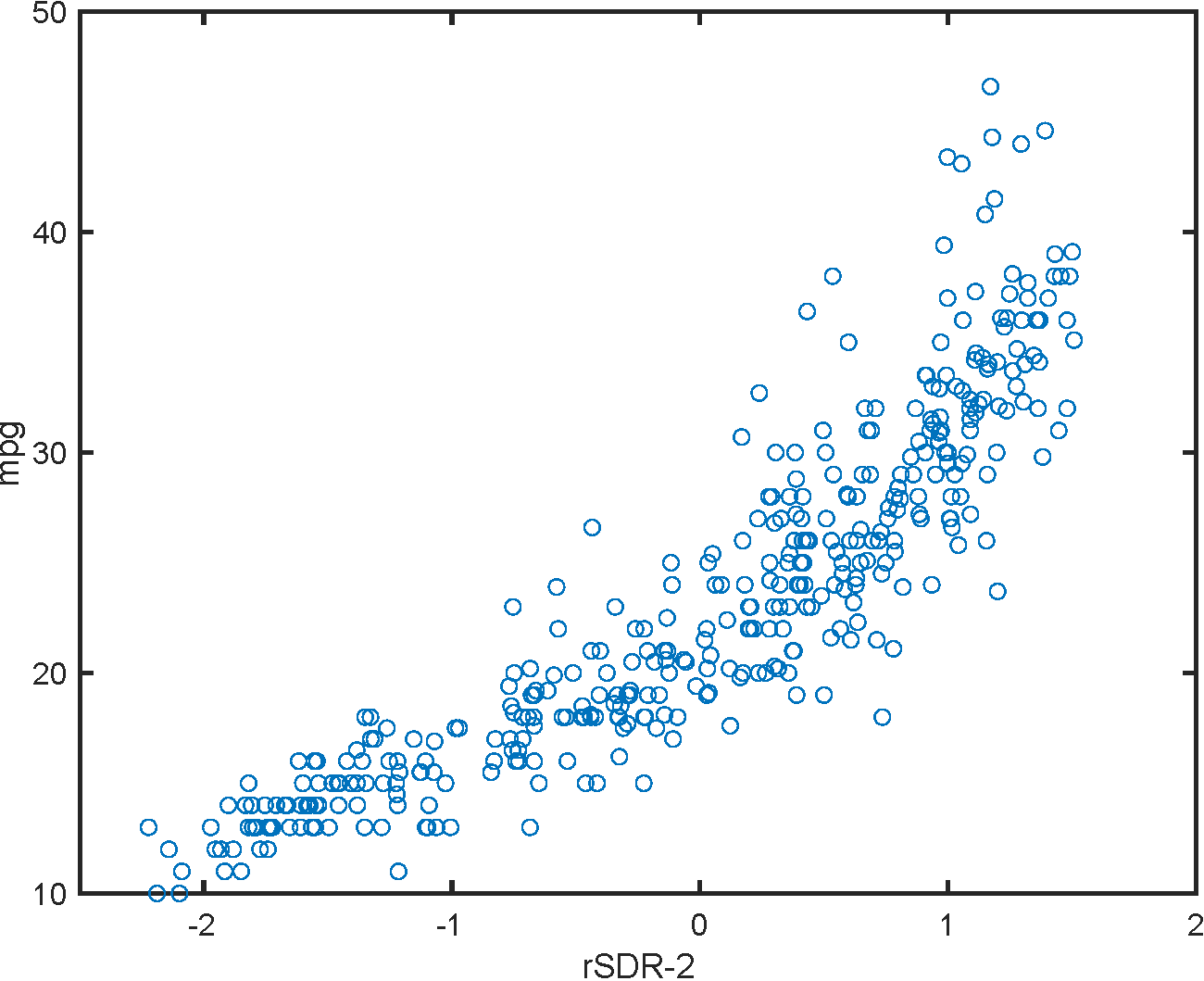}}}\hspace{5pt}
\caption{Scatter plots of ``mpg'' versus the indexes produced by SQP and rSDR ($\alpha=0.2$). }
\label{fig:mpg-plots}
\end{figure}

\begin{table}[t!]
\centering
\renewcommand{\arraystretch}{1.2}
\tbl{Adjusted R-squared and F-value of models from SQP, rSDR in MPG dataset.}{%
\begin{tabular}{ccccc} \toprule
Adjusted R-squared    & rSDR ($\alpha=0.2$) & rSDR ($\alpha=0.5$) & rSDR ($\alpha=1$) \\ \midrule
Linear            & 0.807                & 0.806                & 0.804                \\
Nonlinear         & 0.850                & 0.853                & 0.845                \\ \midrule
F-value              & rSDR ($\alpha=0.2$) & rSDR ($\alpha=0.5$) & rSDR ($\alpha=1$) \\ \midrule
Linear              & 817                  & 816                  & 807                  \\
Nonlinear           & 444                  & 456                  & 427                  \\ \bottomrule
\end{tabular}}
\label{tab:MPG}
\end{table}

\section{Discussion}

In this article, the proposed rSDR using $\alpha$-dCov is robust against outliers in both the response and predictors. Further, the proposed manifold-learning estimation method is less sensitive to the choice of the initial estimators. Both simulation and real-world data applications show that the proposed method outperforms the existing methods. 
The proposed method does not suffer from multicollinearity which could impact the performance of the traditional SDR methods in high-dimensional data analysis. 
 Simulation and real-world data studies show its advantages in terms of computational efficiency and robustness against outliers.

\section*{Acknowledgments}

The authors would like to thank the Editor, the
Associate Editor and the reviewers for their constructive and insightful
comments that greatly improved the manuscript.
This work was
partially supported by NSF grants (DMS-1924792, DMS-2318925 and CNS-1818500).

\section*{Author contribution}

 Hsin-Hsiung Huang: conceptualization, methodology, formal analysis, investigation, 
writing-original draft preparation, writing-review, supervision. 
Feng Yu:
conceptualization, methodology, formal analysis, investigation, writing review, programming and numerical results. 
Teng Zhang: conceptualization, methodology,
formal analysis, investigation, writing-review, supervision. All authors
have read and agreed to the published version of the manuscript.

\section*{Declarations}

The authors declare no competing interests.

\bibliographystyle{apacite}
\bibliography{main_JNS}

\end{document}